\documentclass[letterpaper, 10 pt, conference]{ieeeconf}  

\IEEEoverridecommandlockouts                              

\usepackage{graphics} 
\usepackage{amsmath} 
\usepackage{amssymb}  

\usepackage{graphicx}
\usepackage{xcolor}

\newtheorem{proposition}{Proposition}[section]
\newtheorem{lemma}{Lemma}[section]
\newtheorem{remark}{Remark}[section]

\title{\LARGE \bf
Data--driven feedforward control design for nonlinear systems: \\
A control--oriented system identification approach*
}

\author{Max Bolderman$^{\dagger, 1}$, Mircea Lazar$^1$, and Hans Butler$^{1,2}$
\thanks{*This work is supported by the NWO research project PGN Mechatronics, project number 17973.}
\thanks{$^{\dagger}$Corresponding author: {\tt\small m.bolderman@tue.nl}.}
\thanks{$^{1}$Control Systems Group, Eindhoven University of Technology, Groene Loper 19, Eindhoven, 5612 AP, The Netherlands.}%
\thanks{$^{2}$ASML, De Run 6501, Veldhoven, 5504 DR, The Netherlands.}%
}

\begin{document}

\maketitle
\thispagestyle{empty}
\pagestyle{empty}

\begin{abstract}
Feedforward controllers typically rely on accurately identified inverse models of the system dynamics to achieve high reference tracking performance.
However, the impact of the (inverse) model identification error on the resulting tracking error is only analyzed a posteriori in experiments.
Therefore, in this work, we develop an approach to feedforward control design that aims at minimizing the tracking error a priori. 
To achieve this, we present a model of the system in a lifted space of trajectories, based on which we derive an upperbound on the reference tracking performance.
Minimization of this bound yields a feedforward control--oriented system identification cost function, and a finite--horizon optimization to compute the feedforward control signal.
The nonlinear feedforward control design method is validated using physics--guided neural networks on a nonlinear, nonminimum phase mechatronic example, where it outperforms linear ILC.
\end{abstract}

\section{INTRODUCTION}
\label{sec:Introduction}
Feedforward control is a dominant actor in achieving high reference tracking performance, and typically relies on linear, physics--based models~\cite{Boerlage2003, Devasia2002}.
Linear models have good extrapolation properties, but limited accuracy. 
As such, it would be desirable to employ rich, nonlinear models for feedforward control that can learn the complete system dynamics from data~\cite{Bolderman2023}. 

A common approach to feedforward control design is inverse model--based feedforward, which generates the feedforward signal by passing the reference through a model of the inverse system dynamics, see, e.g.,~\cite{Boerlage2003, Butterworth2012}. 
When the model of the system is nonminimum phase, i.e., it has an unstable inverse, different methods are available to generate a stable feedforward controller, see, e.g.,~\cite{Butterworth2012, Zundert2018}.
These methods however, are not directly extendable to nonlinear feedforward controllers.
Hence, a different approach is to formulate feedforward control as an optimization problem, where the goal is to minimize the norm of the difference between the reference and the model output. 
Within this category, it is possible to optimize the complete feedforward signal~\cite{Volckaert2009, Carrasco2011}, or to parameterize the feedforward signal as a function of time or the reference and optimize over the parameters~\cite{Ramani2017, Kasemsinsup2017}. 
When the system performs a repetitive task, an iterative learning control (ILC) method can be used to minimize the tracking error based on the tracking error of previous repetitions by updating the feedforward input~\cite{Bristow2006}, the parameters of an inverse model~\cite{Blanken2017}, or both~\cite{Saltik2022}. 

The aforementioned methods typically assume a known, physics--based model of the system.
To account for unknown dynamics, data--driven techniques have been explored in combination with artificial intelligence, e.g., neural network (NN) models~\cite{Sorensen1999}, physics--informed neural networks and physics--guided neural network (PGNN) models~\cite{Bolderman2021}, other hybrid model structures~\cite{Chou2023}, or Gaussian processes~\cite{Jilles2022}. 

When performing the identification, i.e., fit the model to the data, the identification cost function should be relevant for the intended use of the model~\cite{Hof1995, Schoukens2019}. 
Therefore, when identifying a model for feedforward control, the identification cost function should push model errors in a region where these errors least affect the tracking performance.
This is not achieved when performing an identification of the inverse dynamics directly, which is generally adopted in nonlinear (including PGNN) feedforward control due to the non--invertibility of nonlinear models in general, see, e.g.,~\cite{Bolderman2023, Bolderman2021, Kon2022}. 
In~\cite{Aarnoudse2021}, a control--relevant identification cost function, which filters the inverse model error with a linear model of the process sensitivity, was proposed to mitigate this issue. 
Alternatively, in~\cite{Bolderman2022c}, the authors proposed an inversion method for (PG)NNs which opens up the path to perform the identification of the forward dynamics, but did not yet achieve a quantitative relation between the tracking error and the identification error.
Such a quantitative relation is desired to have a relevant identification cost function for the feedforward control objective. 

Motivated by the above observations, in this paper, we establish a quantitative link between the tracking error and the identification error. 
The main contributions of this paper are as follows:
\begin{enumerate}
	\item A lifted formulation of the nonlinear feedforward control problem in the space of finite--length trajectories, which enables the derivation of an explicit upperbound on the norm of the tracking error;
	\item A feedforward control--oriented identification cost function, which minimizes the upperbound on the tracking error, and hence, it minimizes the tracking error itself a priori, during the design stage;
	\item A finite--horizon optimal feedforward control (FHOFC) formulation for a general class of nonlinear, possibly nonminimum phase MIMO systems which allows for specifying input, output, and state constraints.
\end{enumerate}
The developed FHOFC problem can be solved iteratively, yielding an iterative learning scheme for nonlinear systems. 
We prove that this iterative learning FHOFC recovers linear ILC~\cite{Bristow2006}, which is another contribution of this work.


\begin{figure}
\centering
\vspace{7pt}
\includegraphics[width=1.0\linewidth]{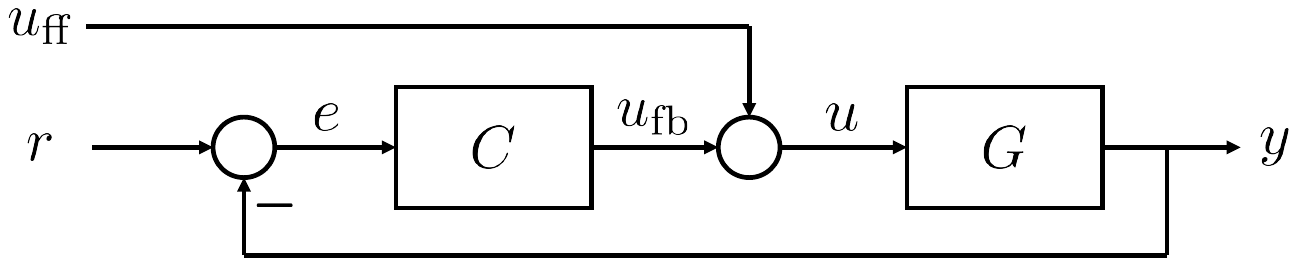}
\caption{Schematic overview of the control structure.}
\label{fig:ClosedLoopSystem}
\end{figure}

\section{PRELIMINARIES}
\label{sec:Preliminaries}
\subsection{Notation}
We denote $y(k) \in \mathbb{R}^{n_y}$ as the output at time $k \in \mathbb{N}_{>0}$, $r(k) \in \mathbb{R}^{n_y}$ is the reference, $e(k) := r(k)-y(k)$ the tracking error and $n_y \in \mathbb{Z}_{>0}$ the number of outputs.
The input $u(k) \in \mathbb{R}^{n_u}$ is the sum of the feedback and the feedforward input, such that $u(k) := u_{\textup{fb}}(k) + u_{\textup{ff}}(k)$ with $n_u \in \mathbb{Z}_{>0}$ the number of inputs. 
The state of a system is denoted as $x(k) \in \mathbb{R}^{n_x}$ with $n_x \in \mathbb{Z}_{>0}$ the state dimension.
A signal of length $N_k \in \mathbb{Z}_{>0}$ is denoted by its capital letter, e.g., $R := [r(1)^T, ..., r(N_k)^T]^T$ is the reference signal and $E := [e(1)^T, ..., e(N_k)^T]^T$ is the error signal.
The superscript $d$ is used to indicate that a signal is from the data set, e.g., $U_{\textup{ff}}^d = [u_{\textup{ff}}^d(1)^T, ..., u_{\textup{ff}}^d(N_d)^T]^T$ is the feedforward input measured during the data generating experiment of length $N_d \in \mathbb{Z}_{>0}$. 
Let a hat denote a prediction of a model, e.g., $\hat{Y} := [\hat{y}(1)^T, ..., \hat{y}(N_k)^T]^T $ is a prediction of the output $Y$ and $\hat{Y}^d := [\hat{y}^d(1)^T, ..., \hat{y}^d(N_d)^T]^T$ a prediction of $Y^d$. 
A model is parametrized by the parameters $\theta \in \mathbb{R}^{n_{\theta}}$, $n_{\theta} \in \mathbb{Z}_{>0}$, and $\hat{\theta}$ denotes the identified parameters. 

\subsection{System dynamics and model--based feedforward}
We consider the feedforward control design for a system~$G$ operating in closed--loop with feedback controller $C$ as visualized in Fig.~\ref{fig:ClosedLoopSystem}. 
The closed--loop system dynamics is (partly) unknown, which is the case in real--life systems, e.g., for a linear motor, one has to deal with parasitic effects, such as nonlinear friction and electromagnetic distortions. The closed--loop system dynamics is denoted as $\phi$, such that
\begin{align}
\begin{split}
\label{eq:SystemDynamics}
	\phi : \begin{cases} x(k+1) & = f \big( x(k), u(k) \big), \\
	y(k) & = g \big( x(k) \big), \\
	u(k) & = C(q) \big( r(k) - y(k) \big) + u_{\textup{ff}}(k). 
	\end{cases}
\end{split}
\end{align}
In~\eqref{eq:SystemDynamics}, $f : \mathbb{R}^{n_x} \times \mathbb{R}^{n_u} \rightarrow \mathbb{R}^{n_x}$ describes the unknown system dynamics, with $g : \mathbb{R}^{n_x} \rightarrow \mathbb{R}^{n_y}$ the unknown output equation. 
The feedback controller is assumed to be linear, such that $u_{\textup{fb}}(k) = C(q) \big( r(k) - y(k) \big)$ with $C(q)$ the discrete--time transfer function of the feedback controller, and $q$ the forward shift operator.

An input--output data set is generated by exciting the system via the reference $r(k)$ and the feedforward input $u_{\textup{ff}}(k)$, such that we obtain $Y^d = [y^d(1)^T, ..., y^d(N_d)^T]^T$, $R^d = [r^d(1)^T, ..., r^d(N_d)^T]^T$, and $U_{\textup{ff}}^d = [u_{\textup{ff}}^d(1)^T, ..., u_{\textup{ff}}^d(N_d)^T]^T$ that satisfies~\eqref{eq:SystemDynamics} for $k = 1, ..., N_d$.

The optimal feedforward input $u_{\textup{ff}}(k)$ yields $y(k) = r(k)$ for all $k$ when supplied to the system~\eqref{eq:SystemDynamics}.
However, since $f$ and $g$ in~\eqref{eq:SystemDynamics} are unknown, it is common practice to parameterize a model $\hat{\phi}$ of the system $\phi$ in~\eqref{eq:SystemDynamics}, such that
\begin{align}
\begin{split}
\label{eq:ModelDynamics}
	\hat{\phi} : \begin{cases} \hat{x}(k+1) & = \hat{f} \big( \theta, \hat{x}(k), \hat{u}(k) \big), \\
	\hat{y}(k) & = \hat{g} \big( \theta, \hat{x}(k) \big), \\
	\hat{u}(k) & = C(q) \big( r(k) - \hat{y}(k) \big) + u_{\textup{ff}}(k). 
	\end{cases}
\end{split}
\end{align}
where $\hat{f}$ and $\hat{g}$ are a model of $f$ and $g$ in~\eqref{eq:SystemDynamics}, respectively, and $\theta \in \mathbb{R}^{n_{\theta}}$ are the free parameters.

The state--space model in~\eqref{eq:ModelDynamics} reduces to the input--output representation used in, e.g.,~\cite{Bolderman2021, Aarnoudse2021} by choosing the state as past inputs and outputs, i.e.,
\begin{align}
\begin{split}
\label{eq:ModelDynamicsInputOutput}
	\hat{y}(k+1) & = \hat{f} \big( \theta, [\hat{y}(k)^T, ..., \hat{y}(k-n_a+1)^T,\\
	& \quad \quad  \hat{u}(k-n_k)^T, ..., \hat{u}(k-n_k-n_b)^T]^T \big), \\
	\hat{u}(k) & = C(q) \big( r(k) - \hat{y}(k) \big) + u_{\textup{ff}}(k),
\end{split}
\end{align}
where $n_a, n_b, n_k \in \mathbb{Z}_{>0}$ denote the order of the dynamics. 
Suppose that there exists an inverse relation $\hat{f}^{-1}$ of $\hat{f}$ in~\eqref{eq:ModelDynamicsInputOutput}, such that, with a slight abuse of notation, we have
\begin{align}
\begin{split}
\label{eq:ModelDynamicsFeedforward}
	\hat{u}(k) & = \hat{f}^{-1} \big( \theta, [\hat{y}(k+n_k+1)^T, ..., \hat{y}(k+n_k-n_a+1)^T, \\
	& \quad \quad \hat{u}(k-1)^T, ..., \hat{u}(k-n_b+1)^T]^T \big).
\end{split}
\end{align}
Then, the inverse model--based feedforward controller is obtained by substitution of $\hat{y}(i) = r(i)$, $i = k+n_k-n_a+1, ..., k+n_k+1$, and $\hat{u}(i) = u_{\textup{ff}}(i)$, $i = k-n_b+1, ..., k$.

\section{PROBLEM FORMULATION}
\label{sec:ProblemStatement}
Since the dynamics~$f$ in~\eqref{eq:SystemDynamics} is unknown, the feedforward control design is based on a model~$\hat{f}$ as in~\eqref{eq:ModelDynamics}. 
This typically yields a two--step feedforward controller design procedure, consisting of an \emph{identification} to fit the model~\eqref{eq:ModelDynamics} to the system~\eqref{eq:SystemDynamics} using the data $\{Y^d, R^d, U_{\textup{ff}}^d \}$, and an \emph{inversion} to find the feedforward input $U_{\textup{ff}}$ for which the output $\hat{Y}$ of the model~\eqref{eq:ModelDynamics} follows the reference $R$.

The identification step aims to find the parameters $\theta = \hat{\theta}$ for the model $\hat{\phi}$ in~\eqref{eq:ModelDynamics} that best fit the data by minimizing a cost function, such that
\begin{equation}
\label{eq:IdentificationCriterion}
	\hat{\theta} = \textup{arg} \min_{\theta} V_{\textup{id}} ( \theta, Y^d, R^d, U_{\textup{ff}}^d ) + \| \Lambda (\theta - \theta^* )\|,
\end{equation}
where $V_{\textup{id}} : \mathbb{R}^{n_{\theta}} \times \mathbb{R}^{n_y N_d} \times \mathbb{R}^{n_y N_d} \times \mathbb{R}^{n_u N_d} \rightarrow \mathbb{R}$ is the identification cost function, and $\Lambda \in \mathbb{R}^{n_{\theta} \times n_{\theta}}$ and $\theta^* \in \mathbb{R}^{n_{\theta}}$ are used for regularization. 
Suppose that $\hat{f}$ is a nonlinear input--output representation as in~\eqref{eq:ModelDynamicsInputOutput}. Then, it is not generally possibly to find an inverse $\hat{f}^{-1}$ as in~\eqref{eq:ModelDynamicsFeedforward}.
A common approach to circumvent this issue is to parametrize a model $\hat{f}^{-1}$ directly, and identify its parameters using, e.g., a one step--ahead direct inverse identification, such that
\begin{align}
\begin{split}
\label{eq:IdentificationInverseDynamics}
	V_{\textup{id}} & (\theta, Y^d, R^d, U_{\textup{ff}}^d ) = \frac{1}{N} \sum_{i = 1}^{N_d} \Big( u^d(k) - \hat{f}^{-1} \big( \theta, x^d(k) \Big)^2, \\
	x^d & (k) = [y^d(k+n_k+1), ..., y^d(k+n_k-n_a+1), \\
	& \quad \quad \quad \quad u^d(k-1), ..., u^d(k-n_b+1)]^T. 
\end{split}
\end{align}
This approach is commonly adopted in literature, see, e.g.,~\cite{Bolderman2023, Kon2022}, but fails to provide a quantitative relation between the identification and the tracking error. 
Note that, the tracking error $e(k) = r(k)-y(k)$ does not even have the same unit as $u^d(k) - \hat{f}^{-1} \big( \theta, x^d (k) \big)$. 
In order to mitigate this issue,~\cite{Aarnoudse2021} proposes to filter the inverse model error by a linear model of the process sensitivity, i.e., minimize $G(q)S(q) \Big( u^d(k) - \hat{f}^{-1} \big( \theta, x^d(k) \big) \Big)$, with $G(q)$ and $S(q)$ the transfer function of the model of the system and the sensitivity, respectively. 
This approach establishes a qualitative relation between the tracking error and the identification cost function, but a quantitative relation is still missing. 

Hence, in this work we will address two main issues: how to design the identification cost function $V_{\textup{id}}$ in~\eqref{eq:IdentificationCriterion}, and how to compute the feedforward control input $u_{\textup{ff}}(k)$ based on the identified nonlinear model, i.e., $\hat{\phi}$ in~\eqref{eq:ModelDynamics} with $\theta = \hat{\theta}$, such that the resulting tracking error is minimized for a general class of nonlinear, possibly nonminimum phase MIMO systems. 

\begin{figure}
\centering
\vspace{7pt}
\includegraphics[width=1.0\linewidth]{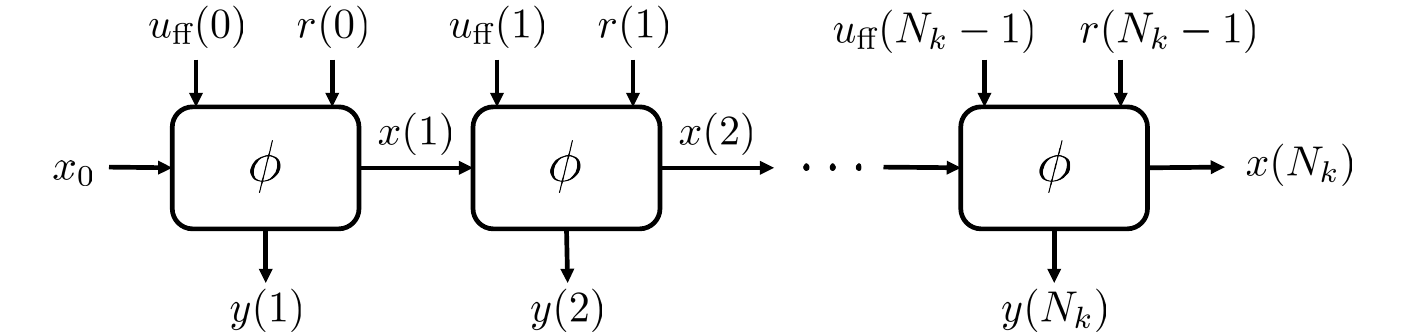}
\caption{Visual representation of the dynamics~\eqref{eq:SystemDynamics} in lifted form~\eqref{eq:SystemDynamicsLifted}.}
\label{fig:Lifted_Form}
\end{figure}

\section{FEEDFORWARD CONTROL--ORIENTED IDENTIFICATION}
\label{sec:IdentificationForControl}
We rewrite the closed--loop system~\eqref{eq:SystemDynamics} and model~\eqref{eq:ModelDynamics} into a lifted form, based on which we propose the feedforward control--oriented identification cost function as well as the FHOFC formulation.
Afterwards, we show that, indeed, minimizing the identification cost function and finding an optimal solution to the FHOFC minimizes the reference tracking error.

The lifted form of the closed--loop system is obtained by simulating~$\phi$ in~\eqref{eq:SystemDynamics} as in Fig.~\ref{fig:Lifted_Form}, such that
\begin{equation}
\label{eq:SystemDynamicsLifted}
	Y = \Phi (x_0, R, U_{\textup{ff}}).
\end{equation}
In~\eqref{eq:SystemDynamics}, $x_0 \in \mathbb{R}^{n_x}$ are the initial conditions, and $\Phi : \mathbb{R}^{n_x} \times \mathbb{R}^{n_y N_k} \times \mathbb{R}^{n_u \times N_k} \rightarrow \mathbb{R}^{n_y N_k}$ is a mapping obtained by recursive composition of the system~$\phi$.
Similarly, the lifted form of the model~$\hat{\phi}$ in~\eqref{eq:ModelDynamics} is defined as follows
\begin{equation}
\label{eq:ModelDynamicsLifted}
	\hat{Y} = \hat{\Phi} (\theta, x_0, R, U_{\textup{ff}} ),
\end{equation}
where $\hat{\Phi} : \mathbb{R}^{n_{\theta}} \times \mathbb{R}^{n_x} \times \mathbb{R}^{n_y N_k} \times \mathbb{R}^{n_u N_k} \rightarrow \mathbb{R}^{n_y N_k}$ is a model--based mapping obtained by recursive composition of~$\hat{\phi}$. 
The data set is generated by exciting the system with $U_{\textup{ff}}^d$ and $R^d$, such that, from~\eqref{eq:SystemDynamicsLifted} and~\eqref{eq:ModelDynamicsLifted}, we can write
\begin{equation}
\label{eq:DataSetLifted}
	Y^d = \Phi (x_0^d, R^d, U_{\textup{ff}}^d), \quad \hat{Y}^d = \hat{\Phi} (\theta, x_0^d, R^d, U_{\textup{ff}}^d),
\end{equation}
with $x_0^d \in \mathbb{R}^{n_{\theta}}$ the initial state of the experiment.

We aim to minimize the $p$--norm of the tracking error, i.e., $\| E \|_p = \| R - Y \|_p$ for some $p \in \mathbb{Z}_{\geq 1}$, while keeping the input, output, and states in the safe sets, i.e., $U \in \mathcal{R}_{U}$, $Y \in \mathcal{R}_{Y}$, and $X \in \mathcal{R}_{x}$.
The feedforward control signal computation is done as follows:
\begin{enumerate}
	\item \textbf{Identify} the optimal set of parameters $\hat{\theta}$ for the model~$\hat{\phi}$ in~\eqref{eq:ModelDynamics} according to~\eqref{eq:IdentificationCriterion} with the feedforward control--oriented identification cost function
\end{enumerate}
	\begin{equation}
	\label{eq:CostFunction_FeedforwardControl}
		V_{\textup{id}} (\theta, Y^d, R^d, U_{\textup{ff}}^d) = \frac{1}{N_d^{1/p}} \| Y^d - \hat{\Phi} (\theta, x_0^d, R, U_{\textup{ff}} ) \|_p.
	\end{equation}
\begin{enumerate}
	\setcounter{enumi}{1}
	\item \textbf{Compute} the feedforward input $U_{\textup{ff}}$ using the identified model, i.e.,~$\hat{\phi}$ with $\theta = \hat{\theta}$, according to the FHOFC
\end{enumerate}
\begin{align}
\begin{split}
\label{eq:FHOFCP}
	U_{\textup{ff}}  = \textup{arg}& \min_{U_{\textup{ff}}} V_{\textup{ff}} (\hat{\theta}, R, U_{\textup{ff}} )  + \| \Gamma U_{\textup{ff}} \|, \\
	\textup{subject to:}\;  U_{\textup{ff}} \in \, &\mathcal{R}_{U_{\textup{ff}}}, \; \hat{U} \in \mathcal{R}_U, \; \hat{Y} \in \mathcal{R}_Y, \; \hat{X} \in \mathcal{R}_X. 
\end{split}
\end{align}
\begin{enumerate}
	\item[] with the FHOFC cost function
\end{enumerate}
\begin{align}
\begin{split}
\label{eq:CostFunction_FHOFCP}
	V_{\textup{ff}} (\hat{\theta}, R, U_{\textup{ff}} ) = \frac{1}{N_k^{1/p}} \| R - \hat{\Phi} (\hat{\theta}, x_0, R, U_{\textup{ff}} ) \|_p. 
\end{split}
\end{align}
\begin{remark} The feedforward control--oriented identification cost function~\eqref{eq:CostFunction_FeedforwardControl} penalizes the closed--loop simulation error of the model~$\hat{\phi}$. 
This is different from the one--step--ahead inverse identification in~\eqref{eq:IdentificationInverseDynamics}, even when it is filtered with the process sensitivity. 
Moreover,~\cite{Bolderman2022c} did not consider the feedback controller in the identification, and was therefore unable to link the tracking and identification error.
\end{remark}

\begin{remark}
Solving the FHOFC optimization~\eqref{eq:FHOFCP} becomes computationally expensive when $N_k$ is large. 
However, the partial derivative of $V_{\textup{ff}}$ with respect to $U_{\textup{ff}}$ is known, such that, e.g., the constraint Gauss--Newton approach in~\cite{Volckaert2009} can be used. 
Several other options to reduce the computational complexity are: $1)$ \emph{parametrize the feedforward signal} using basis functions~\cite{Ramani2017, Kasemsinsup2017}, $2)$ \emph{Parametrize an inverse model} of the system and find its parameters via~\eqref{eq:FHOFCP}, or $3)$ \emph{solve~\eqref{eq:FHOFCP} in a receding horizon manner}. 
\end{remark}
\begin{proposition}
\label{prop:Upperbound}
Consider the system~$\phi$ in~\eqref{eq:SystemDynamics} with lifted form $\Phi$ in~\eqref{eq:SystemDynamicsLifted} and a corresponding parametrized model $\hat{\phi}$ in~\eqref{eq:ModelDynamics} with lifted form $\hat{\Phi}$ in~\eqref{eq:ModelDynamicsLifted}. 
Suppose that $\hat{\theta}$ is identified according to~\eqref{eq:IdentificationCriterion} with $V_{\textup{id}}$ in~\eqref{eq:CostFunction_FeedforwardControl}, that $U_{\textup{ff}}$ is obtained from~\eqref{eq:FHOFCP} with $V_{\textup{ff}}$ in~\eqref{eq:CostFunction_FHOFCP}, and define 
\begin{equation}
\label{eq:Varepsilon}
	\varepsilon := \frac{1}{N_k^{1/p}} \| Y - \hat{Y} \|_p - \frac{1}{N_d^{1/p}} \| Y^d - \hat{Y}^d \|_p.
\end{equation}
Then, the tracking error resulting from $U_{\textup{ff}}$ satisfies
\begin{align}
\begin{split}
\label{eq:UpperboundPerformance}
	\frac{1}{N_k^{1/p}} \| R - Y \|_p \leq V_{\textup{id}} (\hat{\theta}, Y^d, R^d, U_{\textup{ff}}^d) + V_{\textup{ff}} ( \hat{\theta}, R, U_{\textup{ff}} ) + \varepsilon. 
\end{split}
\end{align}
\end{proposition}

\begin{proof}
	From the triangular inequality and~\eqref{eq:Varepsilon}, we have
	\begin{align}
	\begin{split}
	\label{eq:Prop1Step1}
		\|R-Y \|_p & = \| R -\hat{Y} + \hat{Y} - Y \|_p \\
		& \leq \| R - \hat{Y} \|_p + \| Y - \hat{Y} \|_p \\
		\leq \| &R - \hat{Y} \|_p + \frac{N_k^{1/p}}{N_d^{1/p}} \| Y^d - \hat{Y}^d \|_p + N_k^{1/p} \varepsilon .
	\end{split}
	\end{align}
	Dividing both sides by $N_k^{1/p}$ and using $V_{\textup{id}}$ and $V_{\textup{ff}}$ as in~\eqref{eq:CostFunction_FeedforwardControl} and~\eqref{eq:CostFunction_FHOFCP} on the right hand side concludes~\eqref{eq:UpperboundPerformance}. 
\end{proof}

The parameter $\varepsilon$ in~\eqref{eq:Varepsilon} is a measure stating the relevance of the \emph{training data} $\{Y^d, R^d, U_{\textup{ff}}^d\}$ with respect to the \emph{operation data} $\{Y, R, U_{\textup{ff}} \}$, which can be interpreted as the \emph{validation data}.
Accordingly, $\varepsilon > 0$ indicates that the training data does not sufficiently represent the operation data, while $\varepsilon < 0$ indicates that the training data covers all system dynamics that contribute to the tracking error. 
We aim for $\varepsilon = 0$, which is achieved when either:
\begin{enumerate}
	\item \emph{Training data} is the \emph{operation data}, i.e., design $\{ R^d, U_{\textup{ff}}^d \} = \{R, U_{\textup{ff}} \}$. An existing feedforward controller could be used for $U_{\textup{ff}}^d$ as an approximation of~$U_{\textup{ff}}$;
	\item \emph{Consistent parameter identification}, such~that $\Phi(x_0, R, U_{\textup{ff}}) =\hat{\Phi}(\hat{\theta}, x_0, R, U_{\textup{ff}})$ and $\Phi (x_0^d, R^d, U_{\textup{ff}}^d) = \hat{\Phi} (\hat{\theta}, x_0^d, R^d, U_{\textup{ff}}^d)$. This requires standard assumptions from system identification, i.e., the system is in the model set, the data is persistently exciting, and the optimization of~\eqref{eq:IdentificationCriterion} yields a global optimum, see, e.g.,~\cite{Bolderman2022c} for the formalized assumptions. 
\end{enumerate}
Suppose that the operation data has the same length as the data set, i.e., $N_k = N_d = N$ and that $x_0^d = x_0$. Then, with the triangular inequality, we upperbound $\varepsilon$ in~\eqref{eq:Varepsilon} by
\begin{align}
\begin{split}
\label{eq:Varepsilon_Upperbound}
	\varepsilon & \leq \frac{1}{N^{1/p}} \| Y - \hat{Y} - Y^d + \hat{Y}^d \|_p \\
	& \leq \frac{1}{N^{1/p}} \left\| \begin{bmatrix} \gamma_R , \gamma_{U_{\textup{ff}}} \end{bmatrix}  \begin{bmatrix} R - R^d \\ U_{\textup{ff}} - U_{\textup{ff}}^d \end{bmatrix} \right\|_p.
\end{split}
\end{align}
where $\gamma_R = \max \big| \frac{\partial (Y - \hat{Y})}{\partial R} \big|$ and $\gamma_{U_{\textup{ff}}} = \max \big| \frac{\partial (Y - \hat{Y})}{\partial U_{\textup{ff}}} \big|$. 
Since $Y$ and $\hat{Y}$ in~\eqref{eq:Varepsilon_Upperbound} result from a closed--loop simulation, the feedback controller $C(q)$ also affects the upperbound in~\eqref{eq:Varepsilon_Upperbound}. 
Note that, for a linear system, $\frac{\partial Y}{\partial R}$ and $\frac{\partial Y}{\partial U_{\textup{ff}}}$ describe the complementary sensitivity and process sensitivity, respectively. 

\begin{remark} 
Once the feedforward input $U_{\textup{ff}}$ is obtained from~\eqref{eq:FHOFCP}, it is possible to use it for generating new data which results in a smaller $\varepsilon$.
\end{remark}

\section{FINITE--HORIZON OPTIMAL FEEDFORWARD CONTROL}
\label{sec:FHOFC}
Next we show that the FHOFC~\eqref{eq:FHOFCP} can recover some of the state--of--the--art methods for feedforward control, namely: inverse model--based feedforward and linear ILC. 
For simplicity, we neglect the constraints in~\eqref{eq:FHOFCP}. 
\begin{lemma}
\label{le:InversionBasedFeedforward}
	Consider the feedforward control design using an identified input--output model~\eqref{eq:ModelDynamicsInputOutput} for which the inverse relation~\eqref{eq:ModelDynamicsFeedforward} is unique. 
	Suppose that the initial conditions are such that $y(i) = r(i)$ for $i \in \{ n_k-n_a+1, ..., n_k \}$. 
	Then, the inverse model--based feedforward controller 
	\begin{align}
	\begin{split}
	\label{eq:InversionBasedFeedforward}
		u_{\textup{ff}}(k) = f^{-1} \big( \hat{\theta}, [& r^T(k+n_k+1), ..., r^T(k+n_k-n_a+1), \\
		& u_{\textup{ff}}^T(k-1), ..., u_{\textup{ff}}^T(k-n_b+1)]^T \big)
	\end{split}
	\end{align}
	solves the unconstrained FHOFC~\eqref{eq:FHOFCP} with $\Gamma = 0$ in a receding horizon manner with a preview of $n_k+1$.  
\end{lemma}
\begin{proof}
	$u_{\textup{ff}}(k)$ appears first in the predicted output $\hat{y}(k+n_k+1)$, such that $u_{\textup{ff}}(k+i)$, $i \in \mathbb{Z}_{>0}$ does not play a role in the cost function when $\Gamma = 0$. 
	Hence, the FHOFC~\eqref{eq:FHOFCP} with preview $n_k+1$ becomes
	\begin{equation}
	\label{eq:Proof3Step1}
		u_{\textup{ff}}(k) = \textup{arg} \min_{u_{\textup{ff}}(k)} \| r(k+n_k+1) - \hat{y} (k + n_k+1) \|. 
	\end{equation}
	Since the minimum is attained for $r(k+n_k+1) = \hat{y}(k+n_k+1)$, the solution of~\eqref{eq:Proof3Step1} equals~\eqref{eq:InversionBasedFeedforward} when~\eqref{eq:Proof3Step1} is solved sequentially, i.e., for $k = 0, 1, ..., N_k$. 
	Observing that, from~\eqref{eq:ModelDynamicsInputOutput}, $\hat{u}(k) = C(q) \big( r(k)-y(k) \big) + u_{\textup{ff}}(k) = u_{\textup{ff}}(k)$ for $r(i) - \hat{y}(i) = 0$, $i \in \{ 0, ..., k \}$, completes the proof. 
\end{proof}

\begin{remark}
	In~\cite{Bolderman2022c}, the optimization~\eqref{eq:Proof3Step1}, which is a specific case of the FHOFC~\eqref{eq:FHOFCP}, was solved by proposing a specific PGNN structure for which $f^{-1}$ of $f$ in~\eqref{eq:ModelDynamicsInputOutput} is known analytically. Moreover, it discusses the use of a numerical solver when $f^{-1}$ is not known. 
\end{remark}

When the identification in~\eqref{eq:IdentificationCriterion} and the FHOFC~\eqref{eq:FHOFCP} are solved accurately, the tracking error may remain large due to $\varepsilon$ in~\eqref{eq:UpperboundPerformance}. 
To improve performance, it is possible to re--identify $\theta$ based on new data that is generated with $U_{\textup{ff}}$. 
Alternatively, when the system has to perform a repetitive tasks, i.e., has to track a reference $R$ multiple iterations, an iterative learning scheme can be implemented when the same reference is executed repetitively, as is done in ILC~\cite{Bristow2006}. 
Inspired by ILC, we present the iterative learning FHOFC (IL--FHOFC) scheme:
\begin{align}
\begin{split}
\label{eq:NonlinearILCEquivalent}
	\zeta & = \hat{\Phi} (\hat{\theta}, x_0, R, U_{\textup{ff}}^{(i)}) + \alpha E^{(i)}, \\
	U_{\textup{ff}}^{(i+1)} & = Q_m \left[ \textup{arg} \min_{U_{\textup{ff}}} \big\| \zeta - \hat{\Phi} (\hat{\theta}, x_0, R, U_{\textup{ff}}) \big\| + \| \Gamma U_{\textup{ff}} \| \right].
\end{split}
\end{align}
In~\eqref{eq:NonlinearILCEquivalent}, the superscript $(i)$ relates to iteration $i \in \mathbb{Z}_{>0}$ of a signal, $E^{(i)} = R - Y^{(i)}$ is the tracking error since $R$ is constant over the iterations, $\zeta \in \mathbb{R}^{n_u N_k}$ is an auxiliary signal, and $Q_m \in \mathbb{R}^{n_u N_k \times n_u N_k}$ contains the Markov parameters of a robustness filter with transfer function matrix $Q$.

\begin{lemma}
\label{le:FHOFCILCEquivalence}
	Suppose that $G$ is the transfer function matrix of a linear system, $C$ is a feedback controller, and $\hat{G}$ is a model of $G$, with all transfer functions represented in the digital domain. 
	Let $\hat{\Phi}$ be the closed--loop model--based mapping obtained from $C$ and $\hat{G}$, and let $\alpha \in (0, 1]$ be a learning gain and $Q$ a robustness filter. 
	Then, the ILC law
	\begin{align}
	\begin{split}
	\label{eq:LinearILCInput}
	\xi^{(i)}(k) & = u_{\textup{ff}}^{(i)}(k) + \alpha \hat{G}^{-1} ( 1 + \hat{G}C) e^{(i)}(k), \\
	u_{\textup{ff}}^{(i+1)}(k) & = Q \xi^{(i)}(k),
	\end{split}
	\end{align}
	solves the IL--FHOFC problem~\eqref{eq:NonlinearILCEquivalent} with $\Gamma = 0$. 
\end{lemma}
\begin{proof}
	The output predicted by $\hat{G}$ at iteration $i$ is 
	\begin{equation}
	\label{eq:LinearClosedLoop}
		\hat{y}^{(i)}(k) = (1+\hat{G}C)^{-1}\hat{G}Cr(k) + (1+\hat{G}C)^{-1}\hat{G}u_{\textup{ff}}^{(i)}(k).
	\end{equation}
	The output of the system with input $\xi^{(i+1)}(k)$ in~\eqref{eq:LinearILCInput} is
	\begin{align}
	\begin{split}
	\label{eq:Proof4Step1}
		\hat{y}^{(i+1)}(k) = \hat{y}^{(i)}(k) + \alpha e^{(i)}(k).
	\end{split}
	\end{align}
	Placing the time entries of~\eqref{eq:Proof4Step1} in a column, shows that
	\begin{align}
	\begin{split}
	\label{eq:Proof4Step2}
		& \big\| \zeta^{(i+1)} - \hat{\Phi} (\hat{\theta}, x_0, R, [\xi^{(i+1)}(1)^T, ..., \xi^{(i+1)}(N_k)^T]^T ) \big\| = 0. 
	\end{split}
	\end{align}
	Hence, the first line of the linear ILC~\eqref{eq:LinearILCInput} solves the optimization in the IL--FHOFC~\eqref{eq:NonlinearILCEquivalent} for $\Gamma = 0$. 
	Since $Q_m$ describes the robustness filter $Q$ in ILC~\eqref{eq:LinearILCInput}, $U_{\textup{ff}}^{(i+1)}$ in~\eqref{eq:NonlinearILCEquivalent} is the column of $u_{\textup{ff}}^{(i+1)}(k)$ in~\eqref{eq:LinearILCInput}. 
\end{proof}
\begin{remark}
	Both $Q_m$ and $\Gamma$ in the IL--FHOFC~\eqref{eq:NonlinearILCEquivalent} are parameters that affect convergence.
	Proving convergence for specific systems and models will be done in future work. 
\end{remark}


\begin{figure}
\centering
\vspace{7pt}
\includegraphics[width=1.0\linewidth]{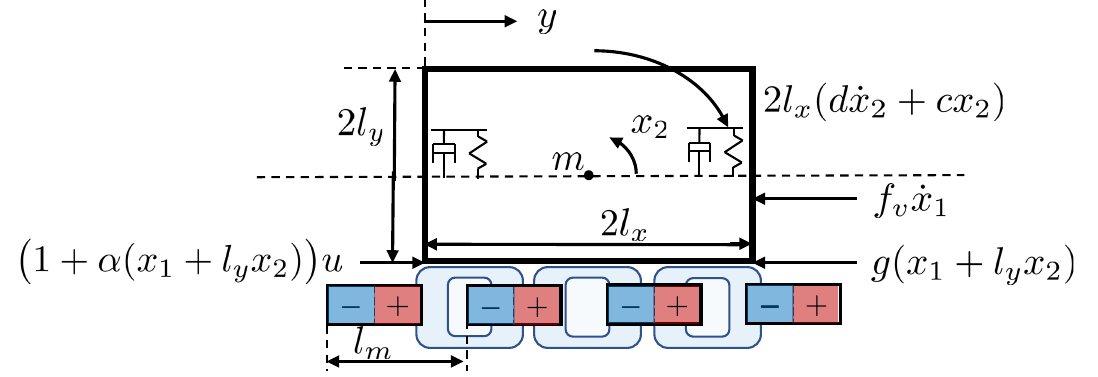}
\caption{Rotating--translating mass with actuation and sensing on opposite sides of the centre of mass.}
\label{fig:NonminimumPhaseExample}
\end{figure}
\begin{table}
\caption{Parameter values of the system displayed in Fig.~\ref{fig:NonminimumPhaseExample}.}
\label{tab:ParameterValues}
\centering
\begin{tabular}{|c|c|c|c|c|c|c|c|c|} \hline
$m$ & $l_x, l_y$ & $J$ & $f_v$ & $c$ & $d$ & $l_m$ & $c_{\alpha}$ & $c_g$ \\ \hline \hline
$20$ & $1$  & $\frac{40}{3}$  & $50$  & $\frac{25 \cdot 10^{3}}{3}$  & $\frac{575}{3}$ & $0.05$ & $0.05$ & $1$ \\ \hline
$kg$ & $m$ & $kgm^2$ & $\frac{kg}{s}$ & $\frac{kg}{s^2}$ & $\frac{kg}{s}$ & $m$ & $-$ & $N$ \\ \hline
\end{tabular}
\end{table}

\section{VALIDATION ON A NONMINIMUM PHASE NONLINEAR MECHATRONIC EXAMPLE}
\label{sec:Results}
\emph{\textbf{System description:}} We consider the position control of a translating--rotating mass with force input $u$ and position output $y$ at opposite sides of the centre of mass, see Fig.~\ref{fig:NonminimumPhaseExample}.
Let $x_1$ be the position of the centre of mass, and $x_2$ the rotation, such that $u$, $x_1$, $x_2$ and $y$ are functions of time. Then, the continuous time dynamics are
\begin{align}
\begin{split}
\label{eq:SimulationModel}
	\ddot{x}_1 & = \frac{1}{m} \Big( -f_v\dot{x}_1 + \big( 1 + \alpha (x_1 + l_y x_2) \big) u \\
	& \quad \quad \quad - g(x_1 + l_y x_2) \Big), \\
	\ddot{x}_2 & = \frac{1}{J} \Big( - 2l_x (d\dot{x}_2	+cx_2) + \\
	& \quad l_y  \big( 1 + \alpha(x_1 + l_y x_2) \big) u - l_y g (x_1 + l_y x_2 ) \Big), \\
	y & = x_1 - l_y x_2. 
\end{split}
\end{align}
In~\eqref{eq:SimulationModel}, $l_x, l_y \in \mathbb{R}_{\geq 0}$ are the width and height of the mass $m \in \mathbb{R}_{>0}$, $J = \frac{1}{3} m(l_x^2 + l_y^2)$ is the moment of inertia, $f_v \in \mathbb{R}_{>0}$ the viscous friction coefficient, and $d, c \in \mathbb{R}_{>0}$ the damping and spring constant counteracting rotation from both ends of the mass. 
The nonlinearities $\alpha (x_1+l_yx_2)$ and $g(x_1 + l_y x_2)$ represent the force ripple and cogging force, and are \emph{unknown}. For simulation purposes, they are modelled as
\begin{align}
\begin{split}
\label{eq:RippleAndCogging}
	\alpha (x_1 + l_y x_2) & = c_{\alpha} \sin \big( \frac{2 \pi}{l_m} (x_1 + l_y x_2) + \frac{\pi}{4} \big), \\
	g (x_1 + l_y x_2) & = c_{g} \sin \big( \frac{2 \pi}{l_m} (x_1 + l_y x_2) \big) ,
\end{split}
\end{align}
with $l_m \in \mathbb{R}_{>0}$ the magnet magnet pole pitch and $c_{\alpha}, c_{g} \in \mathbb{R}$ the riple and cogging magnitude. Parameter values are listed in Table~\ref{tab:ParameterValues}.
The system~\eqref{eq:SimulationModel} is controlled in closed--loop at a frequency $F_s = 100$ $Hz$ by the ZOH discretization of
\begin{equation}
\label{eq:SimulationModelFeedback}
	C(s) = 5\cdot 10^{3} \frac{s+4\pi}{s+20 \pi}.
\end{equation}
A normally distributed noise $v(k) \sim \mathcal{N} \big( 0, (10^{-6})^2 \big)$ $m$ is added as \emph{measurement noise}. 
The system~\eqref{eq:SimulationModel} exhibits the following challenges for feedforward control:
\begin{enumerate}
	\item \emph{Nonlinear dynamics} on the input (force ripple) and output (cogging force), depending on an internal state;
	\item \emph{Nonminimum phase dynamics} which requires non--causal actuation to return a stable feedforward signal.
\end{enumerate}

\emph{\textbf{Data generation:}} The training data is generated in closed--loop by sampling the output $y^d(k)$ at the frequency $F_s$ for a duration of $45$ $s$. 
The reference $R^d$ is a concatenation of $15$ times the third order reference $R$ in Fig.~\ref{fig:Reference}, which has bounded velocity $| \frac{d}{dt} r^d | \leq 0.1$ $\frac{m}{s}$, acceleration $| \frac{d^2}{dt^2} r^d | \leq 4$ $\frac{m}{s^2}$ and jerk $|\frac{d^3}{dt^3} r^d | \leq 40$ $\frac{m}{s^3}$. 
Additionally, $U_{\textup{ff}}^d$ is a zero--mean white noise with variance $\sigma^2 = 10^2$ $N^2$ for $t = [10, 40)$ $s$.

\begin{figure}
\centering
\vspace{7pt}
\includegraphics[width=0.9\linewidth]{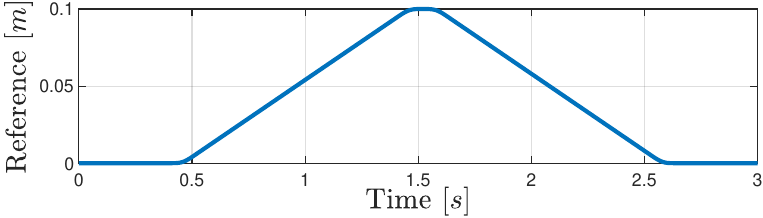}
\caption{Reference $R$ used for performance evaluation.}
\label{fig:Reference}
\end{figure}
\begin{figure}
\centering
\includegraphics[width=1.0\linewidth]{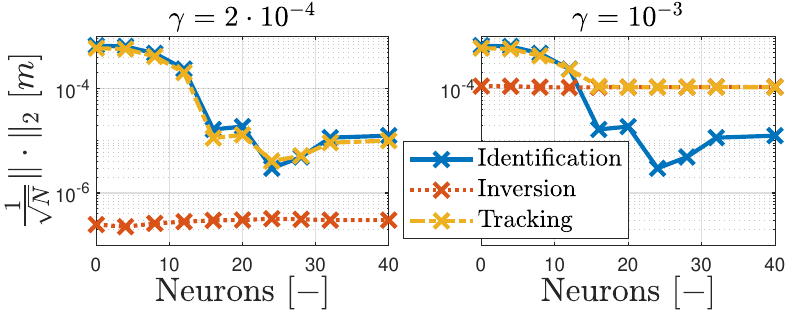}
\caption{Normalized $2$--norm of the identification ($Y^d - \hat{Y}^d$, $N = N_d$), inversion ($R - \hat{Y}$, $N=N_k$) and tracking ($R-Y$, $N = N_k$) error for PGNNs with different number of neurons using $\gamma = 2 \cdot 10^{-4}$ (left window) and $\gamma = 10^{-3}$ (right window).}
\label{fig:TrackingErrorsNeurons}
\end{figure}

\emph{\textbf{Model parametrization:}} We parameterize the system~\eqref{eq:SimulationModel} with a state--space PGNN as
\begin{align}
\begin{split}
\label{eq:StateSpacePGNN}
	\dot{\hat{x}} &= A(\theta_{\textup{phy}}) \hat{x} + B(\theta_{\textup{phy}})\big( \hat{u} + f_{\textup{NN}} (\theta_{\textup{NN}}, \hat{x}, \hat{u}) \big), \\
	\hat{y} & = C(\theta_\textup{phy}) \hat{x}, 
\end{split}
\end{align}
where $\hat{u}$, $\hat{x}$ and $\hat{y}$ are functions of time,~$f_{\textup{NN}} :\mathbb{R}^{n_{\theta_{\textup{nn}}}} \times \mathbb{R}^{4} \times \mathbb{R} \rightarrow \mathbb{R}^4$ is a NN, and $\theta = [\theta_{\textup{phy}}^T, \theta_{\textup{NN}}^T]^T$ with $\theta_{\textup{phy}}$ representing the physical parameters, and $\theta_{\textup{NN}}$ the NN weights and biases. 
The NN has a single hidden layer with $n_1 \in \mathbb{Z}_{\geq 0}$ neurons, which we can vary. 
We discretize~\eqref{eq:StateSpacePGNN} using ZOH while assuming that $f_{\textup{NN}}$ remains constant in between two samples. 


\emph{\textbf{System identification:}} The PGNN parameters are identified according to~\eqref{eq:IdentificationCriterion} with feedforward control--oriented identification cost function~\eqref{eq:CostFunction_FeedforwardControl} (using the \emph{lsqnonlin} MATLAB function) with $\Lambda = 10^{-5} [\textup{diag}(\theta_{\textup{phy}}^T)^{-1} , 0]$, $\theta^* = [{\theta_{\textup{phy}}^*}^T, 0]^T$, and $\theta_{\textup{phy}}^*$ the parameters obtained for the linear part of~\eqref{eq:StateSpacePGNN}. 

\emph{\textbf{Feedforward:}} The feedforward signal $U_{\textup{ff}}$ has length $N_k = 300$ and is computed by solving the FHOFC~\eqref{eq:FHOFCP} while penalizing the rate of change in $U_{\textup{ff}}$ by choosing
	$\Gamma = \gamma \Delta,$
with $\gamma \in \mathbb{R}_{\geq 0}$ the importance of the regularization and $\Delta \in \mathbb{R}^{N_k \times N_k}$ has $1$ on the diagonal, $-1$ on the subdiagonal and $0$ elsewhere.
Solving the FHOFC~\eqref{eq:FHOFCP} converges in $7$ $s$ with \emph{lsqnonlin} on a $2.59$ $GHz$ Intel Core--I7--9750H using MATLAB 2019A.

\emph{\textbf{Results:}} Fig.~\ref{fig:TrackingErrorsNeurons} visualizes the $2$--norm of the identification, inversion and tracking error in~\eqref{eq:UpperboundPerformance} for the reference $R$ in Fig.~\ref{fig:Reference} using different number of neurons $n_1$, with $\gamma = 2\cdot 10^{-4}$ (left window) and $\gamma = 1 \cdot 10^{-3}$ (right window). 
Increasing the number of neurons $n_1$ in the model~\eqref{eq:StateSpacePGNN} improves the accuracy of the identification.
For $\gamma = 2\cdot 10^{-4}$ the inversion error is small, such that the tracking error is limited by the accuracy of the identification.
In contrast, for $\gamma = 10^{-3}$ the inversion error increases, which limits the achievable performance for $n_i > 12$. 
Comparing the upperbound~\eqref{eq:UpperboundPerformance} with the results in Fig.~\ref{fig:TrackingErrorsNeurons} indicates that $\varepsilon$ is small.
Correspondingly, hyperparameters can be tuned based on~\eqref{eq:UpperboundPerformance}, e.g., Fig.~\ref{fig:TrackingErrorsNeurons} shows that at least $n_l = 16$ neurons are required to achieve $\frac{1}{\sqrt{N_k}} \| E \|_2 < 10^{-4}$~$m$. 

Fig.~\ref{fig:ILCImprovements} shows the $2$--norm of the tracking error for a linear and a PGNN model over multiple iterations of the reference $R$ in Fig.~\ref{fig:Reference} when using the IL--FHOFC~\eqref{eq:NonlinearILCEquivalent}.
Since both approaches reach the noise--floor, results are added where $v(k) = 0$ to emphasize the benefit of the nonlinear PGNN model structure. 
The IL--FHOFC~\eqref{eq:NonlinearILCEquivalent} with nonlinear PGNN model yields a significant reduction in the number of iterations required to reach a target performance.  

\begin{figure}
\centering
\vspace{7pt}
\includegraphics[width=1.0\linewidth]{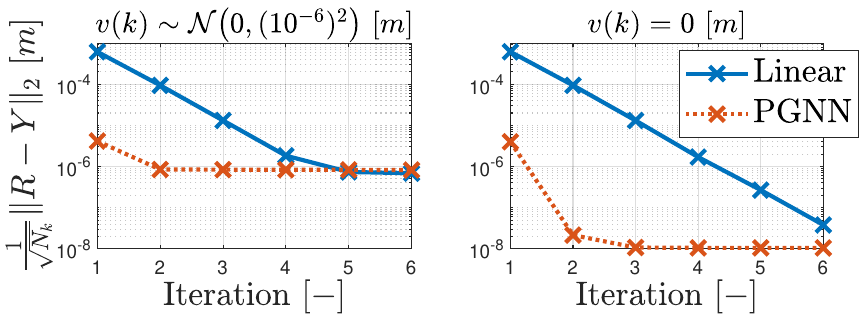}
\caption{Normalized $2$--norm of the tracking error over the iterations using IL--FHOFC~\eqref{eq:NonlinearILCEquivalent} with $\alpha = 1$, $\gamma = 10^{-5}$ and $Q_m = I$ using a linear and a PGNN~\eqref{eq:StateSpacePGNN} model with $n_1 = 24$, simulated with $v(k) \sim \mathcal{N} \big( 0, (10^{-6})^2 \big)$ (left window) and with $v(k) = 0$ (right window).}
\label{fig:ILCImprovements}
\end{figure}

\section{CONCLUSIONS}
\label{sec:Conclusions}
In this paper, we presented a generalized approach to nonlinear data--driven feedforward control design from the perspective of minimizing tracking errors.
We showed that the norm of the reference tracking error is upperbounded by the sum of the inversion and the identification error, respectively. 
This resulted in a two--step approach to feedforward control design, consisting of a feedforward control--oriented system identification followed by a finite--horizon optimization to compute the feedforward input signal.
Generality of the FHOFC formulation was demonstrated by recovering inverse model--based feedforward and linear ILC for specific settings.






\bibliographystyle{IEEEtran}
\bibliography{IEEEabrv,References}

\end{document}